 \documentclass[final,3p,times]{elsarticle}

\usepackage{amsmath}
\usepackage{amsthm}
\usepackage{amssymb}
\usepackage{multirow}
\usepackage{graphicx}
\usepackage{epsfig}
\usepackage{epstopdf}
\usepackage{subfigure}
\usepackage{algorithm}
\usepackage{algorithmic}
\usepackage{setspace}
\usepackage{color}
\usepackage[colorlinks,linkcolor=red,anchorcolor=blue,citecolor=green]{hyperref}

\begin{document}

\begin{frontmatter}



\title{Minimum Enclosing Circle of a Set of Static Points with Dynamic Weight from One Free Point  }

\author{Lei Qiu*, Yu Zhang, Li Zhang}
\address{Department of Electronic Engineering, Tsinghua University, Beijing, China}
\cortext[cor1]{Corresponding author. Email: eeqiulei@gmail.com}
\begin{abstract}
Given a set $S$ of $n$ static points and a free point $p$ in the Euclidean plane, we study a new variation of the minimum enclosing circle problem, in which a dynamic weight that equals to the reciprocal of the distance from the free point $p$ to the undetermined circle center is included. In this work, we prove the optimal solution of the new problem is unique and lies on the boundary of the farthest-point Voronoi diagram of $S$, once $p$ does not coincide with any vertex of the convex hull of $S$. We propose a tree structure constructed from the boundary of the farthest-point Voronoi diagram and use the hierarchical relationship between edges to locate the optimal solution. The plane could be divide into at most $3n-4$ non-overlapping regions. When $p$ lies in one of the regions, the optimal solution locates at one node or lies on the interior of one edge in the boundary of the farthest-point Voronoi diagram. Moreover, we apply the new variation to calculate the maximum displacement of one point $p$ under the condition that the displacements of points in $S$ are restricted in 2D rigid motion.
\end{abstract}

\begin{keyword}
Minimum enclosing circle \sep
Dynamic weight \sep
Farthest-point Voronoi diagram \sep
Center function \sep
Plane division \sep
Division tree


\end{keyword}
\end{frontmatter}

\newtheorem{defn}{Definition}[section]
\newtheorem{Theorem}{Theorem}[section]
\newtheorem{Lemma}{Lemma}[section]
\newtheorem{Observation}{Observation}[section]
\newtheorem{Corollary}{Corollary}[section]
\newtheorem{Fact}{Fact}[section]

\section{Introduction}
The minimum enclosing circle problem is also called the smallest-circle problem or the minimum covering circle problem, which computes the smallest circle that contains all the static points of a given set $S=\{x_1,x_2,...,x_n\}$ in the Euclidean plane. The problem was initially proposed by James J. Sylvester in 1857 \cite{sylvester1857question} and could be solved in linear time \cite{megiddo1983linear,Welzl1991}. Denote the minimum enclosing circle of $S$ by $MEC(S)$, the center of $MEC(S)$ by $\varepsilon(S)$, the radius of $MEC(S)$ by $r(S)$ and the $l_2$ distance between two points $x_i$ and $x_j$ in $\mathbb{R}^2$ by $\|x_i-x_j\|$. Then $\varepsilon(S)$ is the optimal solution of the following min-max problem,
\begin{equation*}
\quad \min_{x}\max_{i}\|x-x_i\|, 
\end{equation*}
and $r(S)$ is the maximum distance from $\varepsilon(S)$ to points in $S$. One variation of the minimum enclosing circle problem is to assign a weight $w_i$ for the distance from the undetermined circle to each static point $x_i$ in $S$ \cite{dearing2013dual},
\begin{equation*}
\quad \min_{x}\max_{i}w_i\|x-x_i\|.
\end{equation*}
The optimal solution of the weighted minimum enclosing circle problem is existing and also unique. The center lies in the convex hull of $S$ \cite{hearn1982efficient} and could also be calculated in linear time \cite{megiddo1983weighted}. Some other variations of the minimum enclosing problem also appeared in literature. Aggarwal et al. \cite{Aggarwal1989} and Chazelle et al. \cite{chazelle1996linear} extended the linear algorithm to cases in higher space $\mathbb{R}^d$ $(d>2)$. Banik et al. \cite{Banik2014} investigated the locus of $\varepsilon(S)$ when one point in $S$ moves along a line.

In this paper we study another variation of the minimum enclosing circle problem. Instead of assigning a fixed weight $w_i$ for the distance term, we assign a dynamic weight, which equals to the reciprocal of the distance from the undetermined center $x$ to another free point $p$ in the Euclidean plane. Denote the dynamic weight by $w(x)$, then we have $w(x)=1/\|x-p\|$, thus the new variation could be represented as,
\begin{equation*}
M(W_0):\quad \min_{x}\max_{i}w(x)\|x-x_i\|,
\end{equation*}
The free point $p$ used in the weight is called as the weight point. When $x\rightarrow p$, $w(x)\rightarrow\infty$. Thus $p$ is a singular point of the objective function in problem $M(W_0)$.

It is easy to get the optimal solution of problem $M(W_0)$ if there is only one point $x_1$ in $S$. The optimal solution is $x_1$ if the weight point does not coincide with $x_1$. In the following article, we will only study the cases in which all points in $S$ are distinct and there are no less than two points in $S$. For such cases, $\max_{i}\|x-x_i\|>0$ and problem $M(W_0)$ has the following equivalent representation,
\begin{equation}
M(W_1):\quad \max_{x}\frac{\|x-p\|}{\max_{i}\|x-x_i\|}.
\end{equation}
$p$ is not a singular point in problem $M(W_1)$, and it is not the optimal solution because $\|x-p\|=0$ when $x=p$.

In the following article we will investigate the relationship between the weight point $p$ and the optimal solution of problem $M(W_1)$ given the static point set $S$.

\section{Preliminaries}

\begin{figure}
  \centering
  \includegraphics[width=0.4\linewidth]{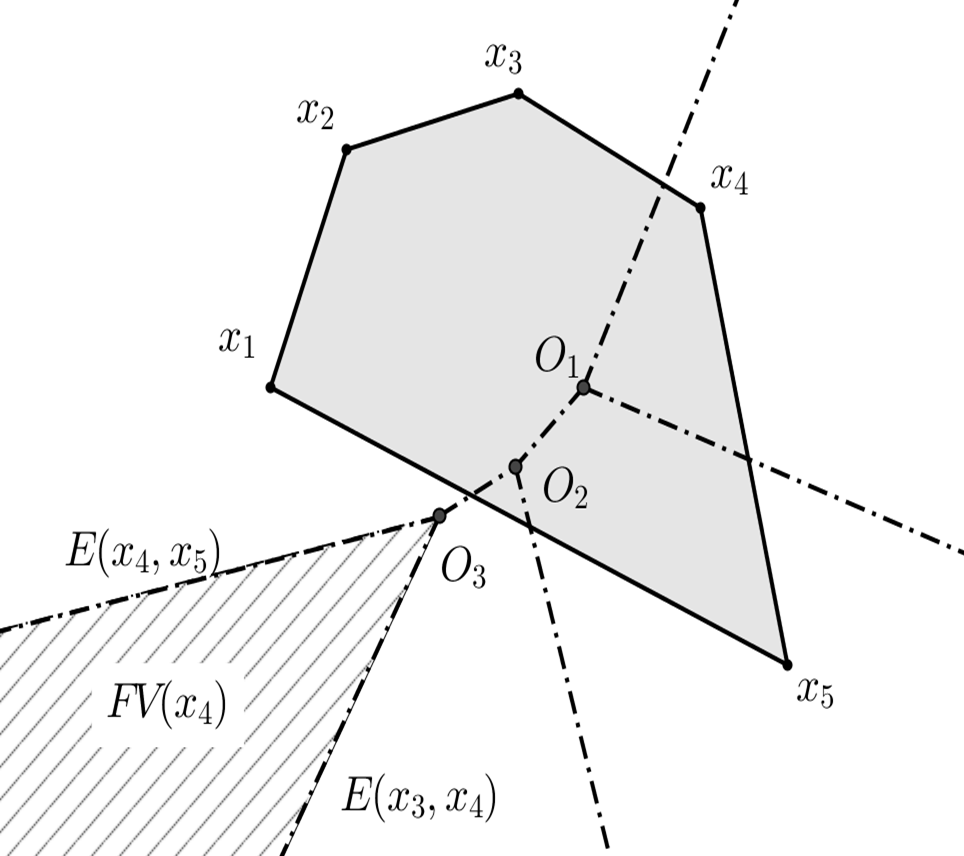}
  \caption{Denotations on the farthest point Voronoi Diagram.}
  \label{pic_fvd}
\end{figure}
Since the optimal solution of problem $M(W_1)$ is related with the farthest-point Voronoi diagram of $S$, we begin by introducing the relevant definitions and denotations. The farthest-point Voronoi diagram of $S=\{x_1,x_2,...,x_n\}$ is a subdivision of the plane into several cells, with property that a point $y$ lies in the cell corresponding to a site $x_i$ if and only if $\|y-x_i\|>\|y-x_j\|$ for each $x_j\in S$ with $j\neq i$ \cite{book2008CG}. The cell corresponding to site $x_i$ is denoted by $FV(x_i)$. The boundary of the farthest-point Voronoi diagram of $S$ is denoted by $FVB(S)$. The boundary between two cells $FV(x_i)$ and $FV(x_j)$ is an edge in $FVB(S)$, which is denoted by $E(x_i,x_j)$. The intersection of edges is a node in $FVB(S)$, which is denoted by $O$. We call the remaining part of $E(x_i,x_j)$ after removing two endpoints as the interior of $E(x_i,x_j)$, which is denoted by $E^\circ(x_i,x_j)$. The convex hull of $S$ is denoted by $CH(S)$. Fig. \ref{pic_fvd} shows a case in which all five static points in $S$ are the vertex of $CH(S)$. In all figures of this article, $FVB(S)$ is shown with the dot dash lines.

Some important properties of the farthest Voronoi diagram are summarized as below:

\begin{Fact}
\label{fact_mec_center}
$\varepsilon(S)$ is a point in $FVB(S)$.
\end{Fact}

\begin{Fact}
\label{fact_outsidecell}
Static point $x_i$ does not belong to its corresponding cell $FV(x_i)$.
\end{Fact}

\begin{Fact}
\label{fact_cell_covex}
$FV(x_i)$ is a convex set.
\end{Fact}

\begin{Fact}
\label{fact_cell_not_empty}
$FV(x_i)$ is not empty if and only if $x_i$ is a vertex of $CH(S)$.
\end{Fact}

\begin{Fact}
\label{fact_fvbs_chs}
$FVB(S)$ is determined by $CH(S)$.
\end{Fact}

\begin{Fact}
\label{fact_midperpendicular}
$E(x_i,x_j)$ lies on the midperpendicular of segment $x_ix_j$. $E^\circ(x_i,x_j)$ and segment $x_ix_j$ intersects at $\varepsilon(S)$ if and only if $\varepsilon(S)\in E^\circ(x_i,x_j)$.
\end{Fact}

\begin{Fact}
\label{fact_infty}
$\infty$ is a special node in $FVB(S)$.
\end{Fact}

A set $S=\{x_1,x_2,...,x_n\}$ of $n$ distinct points in the plane is said to be in general position if no four points of $S$ are co-circular. Only point set in general position is studied in this article.
\begin{Fact}
\label{fact_general_position}
If $S$ is in general position, all nodes in $FVB(S)$ except $\infty$ are the intersection of three adjacent edges.
\end{Fact}
\begin{Fact}
\label{fact_adjointcircle}
Assume in $FVB(S)$ three edges $E_1(x_i,x_j)$, $E_2(x_i,x_k)$ and $E_3(x_k,x_j)$ intersect at node $O$, then $O$ is the center of the circle which passes through the three static points $x_i$, $x_j$ and $x_k$.
\end{Fact}

\section{The restricted feasible set}
The feasible set of problem $M(W_1)$ is the whole Euclidean plane. In this section, we will prove that the optimal solution of problem $M(W_1)$ must lie on $FVB(S)$ once the weight point satisfies an ordinary condition. Thus we just need to locate the optimal solution in a subset of the plane. 

\label{section_location}
\begin{figure}
  \centering
  \includegraphics[width=0.4\linewidth]{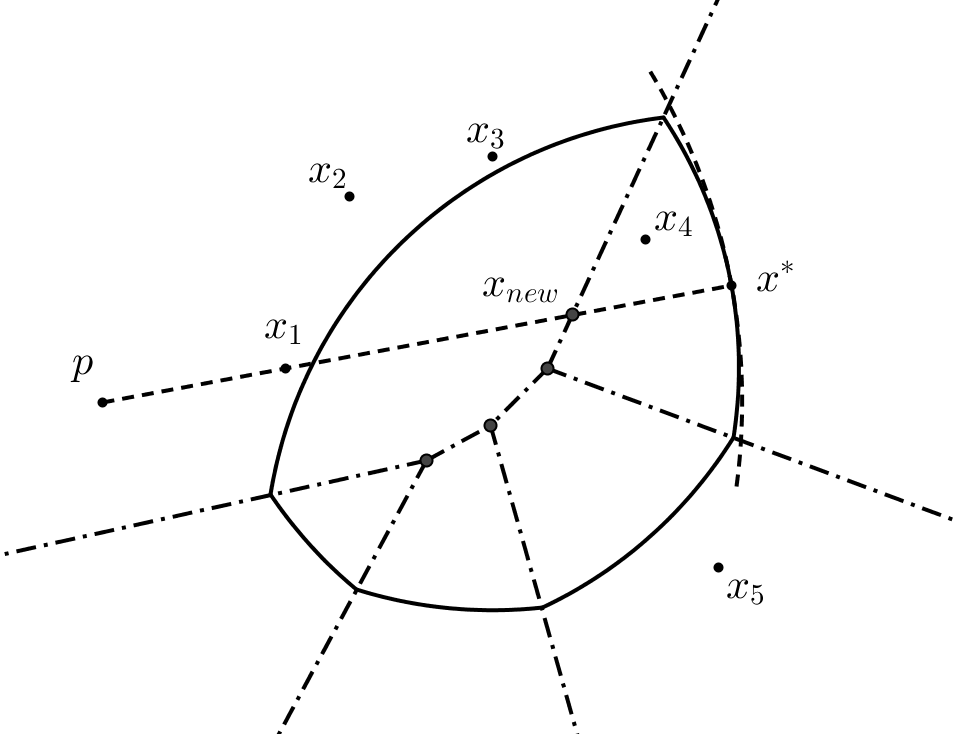}
  \caption{A case of $U(r)$ when $r>r(S)$ and illustration for the proof of Theorem \ref{theorem_location}.}
  \label{pic_proof_location_optimal_solution}
\end{figure}
\begin{defn}
$U(r)$ is a subset of the Euclidean plane. For any point $x$ in $U(r)$, the maximum distance from $x$ to points in $S$ is equal to $r$,
\begin{equation*}
U(r) = \{x|\max_{i}\|x-x_i\|=r,\ x_i\in S\}.
\end{equation*}
\end{defn}
According to the definition of $U(r)$, problem $M(W_1)$ could be represented as,
\begin{equation}
M(W_2):\quad \max_{r}\quad\frac{1}{r}\max_{x\in U(r)}\|x-p\|,
\end{equation}
where $\max_{x\in U(r)}\|x-p\|$ is a subproblem of finding from $U(r)$ the point farthest to the weight point $p$.

\begin{Theorem}
\label{theorem_location}
The optimal solution of $M(W_2)$ lies on $FVB(S)$ if $p$ does not coincide with any vertex of $CH(S)$.
\end{Theorem}
\begin{proof}
(Prove by contradiction.) Fig. \ref{pic_proof_location_optimal_solution} shows an illustration for the proof. As $r(S)=\min_x \max_i \|x-x_i\|$, $U(r)$ is an empty set if $r<r(S)$ and $U(r)=\{\varepsilon(S)\}$ if $r=r(S)$. Since $\varepsilon(S)$ lies in $FVB(S)$ (Fact \ref{fact_mec_center}), only the cases in which $r>r(S)$ need to be considered.
If $r>r(S)$, $U(r)$ consists of a group of connected arcs. Static point $x_i$ is the center of the arc lying in $FV(x_i)$ and the intersection of two connected arcs lies on $FVB(S)$. Supposing the optimal solution $x^*$ of problem $M(W_2)$ does not lie on $FVB(S)$, it must lie inside some nonempty cell of the farthest-point Voronoi diagram of $S$, for example $FV(x_1)$. Denote the arc in $FV(x_1)$ which centers at $x_1$ and passes through $x^*$ by $arc(x_1,x^*)$. Then $x^*$ is also the optimal solution of the subproblem $max_{x\in arc(x_1,x^*)}\|x-p\|$. To guarantee $x^*$ is the farthest point to $p$ in $arc(x_1,x^*)$, $arc(p,x^*)$ and $arc(x_1,x^*)$ should be tangent at $x^*$, and $x_1$ should lie on segment $x^*p$. Since $x_1$ lies outside $FV(x_1)$ (Fact \ref{fact_outsidecell}) and $x^*$ lies inside $FV(x_1)$, segment $x_1x^*$ intersects with the boundary of the convex set $FV(x_1)$ (Fact \ref{fact_cell_covex}). Denote the intersection point by $x_{new}$, then $x_{new}$ is a point in $FVB(S)$. Since $p$ lies outside segment $x_1x^*$, we have $\|x_{new}-p\|/\|x_{new}-x_1\|\geq\|x^*-p\|/\|x^*-x_1\|$, and the equality holds if and only if $p=x_1$. Thus if $p$ does not coincide with $x_1$, there exists a point $x_{new}$ in $FVB(S)$ with larger objective function value than $x^*$ . Because $FV(x_1)$ is not empty if and only if $x_1$ is a vertex of $CH(S)$ (Fact \ref{fact_cell_not_empty}), points not in $FVB(S)$ could not be the optimal solution if $p$ does not coincide with any vertex of $CH(S)$.
\end{proof}

If $p$ coincides with vertex $x_i$ of $CH(S)$, the optimal solution is any point in cell $FV(x_i)$. 
According to Theorem \ref{theorem_location}, problem $M(W_1)$ and $M(W_2)$ is equivalent to the following one when the weight point does not coincide with any vertex of $CH(S)$,
\begin{equation}
M(W_3):\quad \max_{x}\frac{\|x-p\|}{\max_{i}\|x-x_i\|}, \quad x\in FVB(S).
\end{equation}
The feasible set of problem $M(W_3)$ is $FVB(S)$ instead of the whole Euclidean plane in problem $M(W_2)$. $FVB(S)$ is a combination of edges. Thus in order to calculate the optimal solution in $FVB(S)$, we could first calculate the optimal solutions on each single edge $E(x_i,x_j)$ of $FVB(S)$ and then choose the one with the largest objective function value.
\begin{equation}
M(W_4):\quad \max_{E(x_i,x_j)\subset FVB(S)}\max_{x\in E(x_i,x_j)}\frac{\|x-p\|}{\max_{k}\|x-x_k\|}.
\end{equation}
Since $FVB(S)$ is determined by $CH(S)$ (Fact \ref{fact_fvbs_chs}), we assume all points in $S$ are the vertex of $CH(S)$ in the following article.
\section{Optimal solution on a single edge of $FVB(S)$}
\label{section_os_single_edge}
In this section, we investigate the optimal solution on a single edge of $FVB(S)$, which is a subproblem of problem $M(W_4)$,
\begin{equation}
M_s(W_4):\quad \max_x\frac{\|x-p\|}{\max_{k}\|x-x_k\|}, \quad x\in E(x_i,x_j).
\end{equation}
Since $E(x_i,x_j)$ lies on the midperpendicular of segment $x_ix_j$ (Fact \ref{fact_midperpendicular}) and $E(x_i,x_j)$ is the boundary between cell $FV(x_i)$ and $FV(x_j)$, the maximum distance from $x$ to points in set $S$ is $\|x-x_i\|$, or equivalently $\|x-x_j\|$. The objective function of subproblem $M_s(W_4)$ is,
\begin{equation*}
f(x) = \frac{\|x-p\|}{\|x-x_i\|}, \quad x\in E(x_i,x_j).
\end{equation*}

From Fact \ref{fact_midperpendicular}, $E^\circ(x_i,x_j)$ and segment $x_ix_j$ intersect at $\varepsilon(S)$ if and only if $\varepsilon(S)\in E^\circ(x_i,x_j)$. Thus there is at most one edge in $FVB(S)$ satisfying $\varepsilon(S)\in E^\circ(x_i,x_j)$. If there does exist one such edge, we divide the edge into two parts from $\varepsilon(S)$ and treat $\varepsilon(S)$ as a node in $FVB(S)$. After the above process, all edges in $FVB(S)$ lie on one side of the line passing through $x_i$ and $x_j$. Edges in $FVB(S)$ could be classified into two categories, the unbounded edge and the bounded edge. Both two kinds of edges start from one node $O_1$ and end at another node $O_2$ ($O_2=\infty$ if the edge is half-unbounded). $O_1$ is nearer to segment $x_ix_j$ than $O_2$. Segment $x_ix_j$ intersects with $E(x_i,x_j)$ at the starting node $O_1$ or there is no intersection between segment $x_ix_j$ and $E(x_i,x_j)$. The localization of the optimal solution on two kinds of edges will be discussed successively.

\subsection{Unbounded edge}
\begin{figure*}
  \centering
  \includegraphics[width=0.3\linewidth]{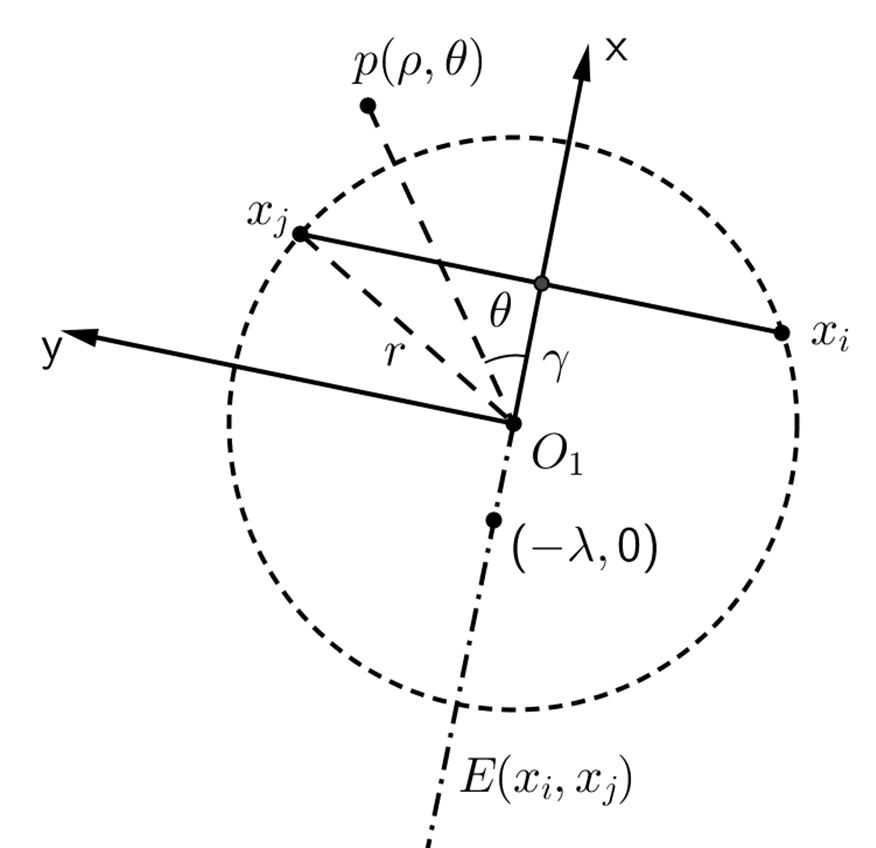}
  \caption{The coordinate system for a single edge in $FVB(S)$}
  \label{pic_unbounded_edge}
\end{figure*}
As shown in Fig. \ref{pic_unbounded_edge}, $E(x_i,x_j)$ is an unbounded edge starting from node $O_1$. The distance from $O_1$ to segment $x_ix_j$ is $\gamma$. The distance from $O_1$ to $x_i$ (or $x_j$) is $r$. Choose a coordinate system such that $O_1$ is the origin and $E(x_i,x_j)$ coincides with the negative x-axis. Denote the feasible solution on $E(x_i,x_j)$ by $x=(-\lambda,0)$ and the weight point by its polar coordinate $p=(\rho,\theta)$. The objective function is,
\begin{equation}
\label{eq_objectfunction_unbounded}
f(\lambda)=\sqrt{\frac{\rho^2+2\rho\lambda\cos\theta+\lambda^2}{r^2+2\gamma\lambda+\lambda^2}} \quad(\lambda\geq0).
\end{equation}
To obtain the optimal solution, we calculate the derivative of $f(\lambda)$,
\begin{equation*}
\frac{\partial f}{\partial \lambda} = \frac{1}{(r^2+2\gamma\lambda+\lambda^2)f(\lambda)}
[(\gamma-\rho\cos\theta)\lambda^2+(r^2-\rho^2)\lambda+(\rho r^2\cos\theta-\rho^2\gamma)].
\end{equation*}
The sign of the derivative depends on the quadratic polynomial in bracket. For fixed parameters $\gamma$ and $r$, the signs of three coefficients of the quadratic polynomial $a=\gamma-\rho\cos\theta$, $b=r^2-\rho^2$ and $c=\rho r^2\cos\theta-\rho^2\gamma$ are determined by the location of the weight point $p$. Equation $a=0$ is a line passing through $x_i$ and $x_j$. Equation $b=0$ represents a circle centered at node $O_1$ and passing through $x_i$ and $x_j$.

If $O_1$ does not lie on segment $x_ix_j$ ($\gamma>0$), equation $c=0$ represents a circle passing through $x_i$, $x_j$ and $O_1$. The circle center is $(r^2/2\gamma,0)$. As shown in Fig. \ref{pic_division_unbounded}(a), three boundaries $a=0$, $b=0$ and $c=0$ divide the plane into six non-overlapping regions. When the weight point $p$ lies in one of the six regions, the coefficient signs, the changes of the objective function and the optimal solutions are listed in Table \ref{table_trend_unbounded}. As the interior of each region is considered in this step, there are no equal signs in the table. When $p$ lies in region $\large{\textcircled{\small{2}}}$, $\large{\textcircled{\small{3}}}$, $\large{\textcircled{\small{4}}}$ or $\large{\textcircled{\small{5}}}$, there exists one positive $\lambda^*$ satisfying $f'(\lambda^*)=0$,
\begin{equation}
\label{eq_lambda}
\lambda^*=-\frac{b}{2a}+\sqrt{\frac{b^2-4ac}{4a^2}}.
\end{equation}
The function values at two endpoints of the unbounded edge are $f(0)=\rho/r$ and $f(\infty)=1$. When $p$ lies in region $\large{\textcircled{\small{2}}}$, $b=r^2-\rho^2<0$, thus $r<\rho$ and $f(0)>f(\infty)$. When $p$ lies in region $\large{\textcircled{\small{4}}}$, $b=r^2-\rho^2>0$, thus $r>\rho$ and $f(0)<f(\infty)$.
From the table we could see different regions may correspond to the same optimal solution, thus the above six regions could be degenerated into three. As shown in Fig. \ref{pic_division_unbounded}(b), the boundaries of the three regions are two arcs $c=0\wedge b<0$, $b=0\wedge a>0$ and one segment $a=0\wedge b>0$. We also use $E:c=0$, $E:b=0$, $E:a=0$ to denote the three boundaries, in which the prefix indicates the edge the boundary associates with.

When $\gamma\rightarrow 0$, we have $r^2/2\gamma\rightarrow\infty$. Thus boundary $E:c=0$ lies on the line passing through $x_i$ and $x_j$ if $O_1$ lies on segment $x_ix_j$ ($\gamma=0$). As shown in Fig. \ref{pic_division_unbounded}(c), the plane still could be divided into three regions and the optimal solutions of each region correspond to either one endpoint or the interior of the unbounded edge.

\begin{figure}
  \centering
  \includegraphics[width=0.9\linewidth]{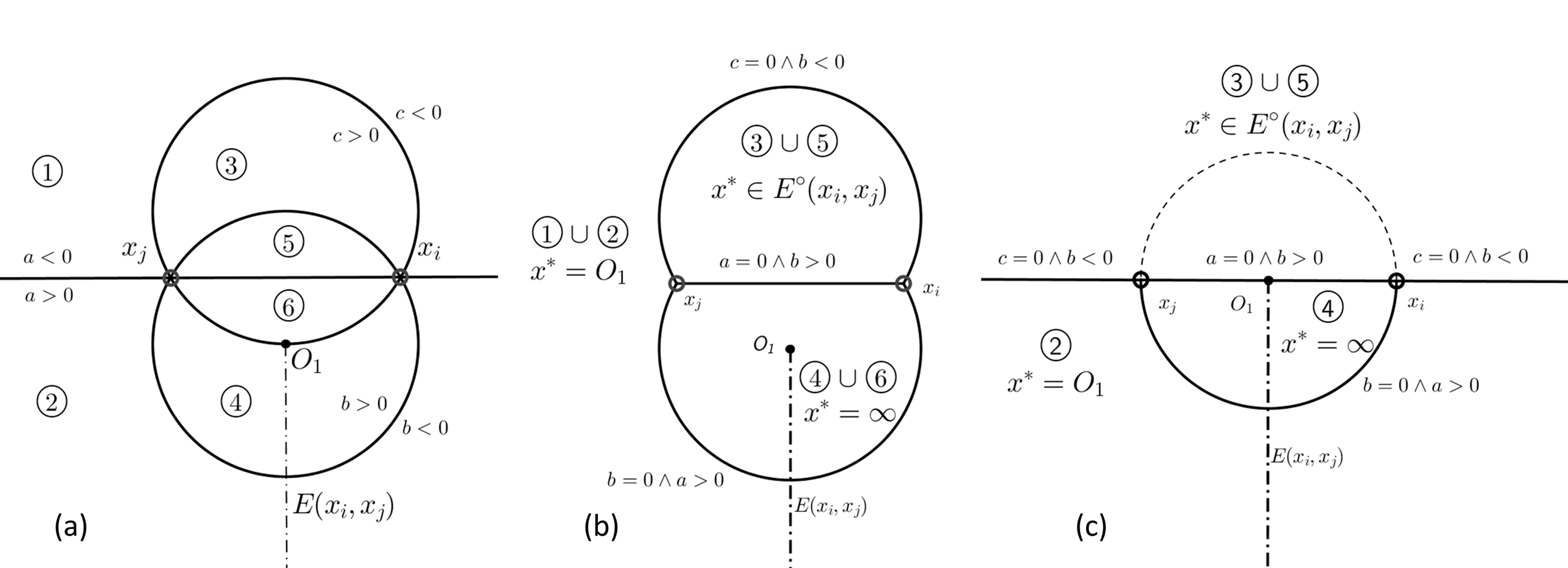}
  \caption{Plane division for a single unbounded edge. (a) Original, $\gamma>0$, (b) Degenerate, $\gamma>0$, (c) Degenerate, $\gamma=0$.}
  \label{pic_division_unbounded}
\end{figure}

\begin{table}
  \caption{Optimal solution on a single unbounded edge}
  \label{table_trend_unbounded}
  \renewcommand{\arraystretch}{1.5}
  \centering
  \small
  \begin{tabular}{|*{6}{c|}}
  \hline
  Region & Coefficient Signs   & $f(\lambda)$   & Value Comparison  & Optimal Value & Optimal Solution \\ \hline
  $\large{\textcircled{\small{1}}}$ & $a<0,\ b<0,\ c<0$ & $(0,+\infty)\searrow$ & $f(0)>f(\infty)$ & $f(0)$ & $O_1$\\ \hline
  $\large{\textcircled{\small{2}}}$ & $a>0,\ b<0,\ c<0$ &$(0,\lambda^*)\searrow(\lambda^*,+\infty)\nearrow$ & $f(0)>f(\infty)$ & $f(0)$ & $O_1$\\ \hline
  $\large{\textcircled{\small{3}}}\cup\large{\textcircled{\small{5}}}$ & $a<0,\ c>0$ &$(0,\lambda^*)\nearrow(\lambda^*,+\infty)\searrow$ & $f(\lambda^*)> \max(f(0),f(\infty))$ & $f(\lambda^*)$ & $\in E^\circ(x_i,x_j)$ \\ \hline
  $\large{\textcircled{\small{4}}}$ &$a>0,\ b>0,\ c<0$&$(0,\lambda^*)\searrow(\lambda^*,+\infty)\nearrow$ & $f(0)<f(\infty)$ & $f(\infty)$ & $\infty$ \\ \hline
  $\large{\textcircled{\small{6}}}$ &$a>0,\ b>0,\ c>0$ &$(0,+\infty)\nearrow$ & $f(0)<f(\infty)$ & $f(\infty)$ & $\infty$ \\ \hline
  \end{tabular}
\end{table}

\subsection{Bounded edge}
\label{subsection_bounded_edge}
Denote the length of the bounded edge by $\delta$ and choose the same coordinate system as the case of unbounded edge, the objective function is the same with Eq. \eqref{eq_objectfunction_unbounded} except that the function domain changes from $\lambda\geq 0$ to $0\leq\lambda\leq\sigma$. Based on the analysis of the previous subsection, the optimal solution on a bounded edge is studied below.

Similarly, we first study the cases in which $O_1$ does not lie on segment $x_ix_j$ ($\gamma>0$). When $p$ lies in region $\large{\textcircled{\small{1}}}$, $f(\lambda)$ decreases monotonically in $(0,\delta)$, thus the optimal solution is $O_1$. When $p$ lies in region $\large{\textcircled{\small{6}}}$, $f(\lambda)$ increases monotonically in $(0,\delta)$, thus the optimal solution is $O_2$.

When $p$ lies in region $\large{\textcircled{\small{3}}}\cup\large{\textcircled{\small{5}}}$, $f(\lambda)$ increases in $(0,\lambda^*)$ and then decreases in $(\lambda^*,+\infty)$. Thus if $\lambda^*\geq\delta$, $f(\lambda)$ increases monotonically in $(0,\delta)$, and if $\lambda^*<\delta$, $f(\lambda)$ increases in $(0,\lambda^*)$ then decreases in $(\lambda^*,\delta)$. From Eq. \eqref{eq_lambda}, $\lambda^*$ is determined by the location of the weight point $p$ given parameter $\gamma$ and $r$, therefore region $\large{\textcircled{\small{3}}}\cup\large{\textcircled{\small{5}}}$ can be subdivided by a new boundary $\lambda^*=\delta$. Since $f'(\lambda^*)=0$, we have $f'(\delta)=0$, so the new boundary satisfies
\begin{equation*}
(\gamma-\rho\cos\theta)\delta^2+(r^2-\rho^2)\delta+(\rho r^2\cos\theta-\rho^2\gamma)=0.
\end{equation*}
Substitute the polar coordinate $(\rho,\theta)$ of the weight point $p$ with the Cartesian one $(p_x,p_y)$, then we have $\rho\cos\theta=p_x$ and $\rho^2=p_x^2+p_y^2$. The above equation could be expressed as
\begin{equation}
\label{eq_boundary1}
\quad (\gamma+\delta)(p_x^2+p_y^2) - (r^2-\delta^2)p_x - (\gamma\delta+r^2)\delta = 0,
\end{equation}
which represents a circle passing through $x_i$, $x_j$ and centered at $((r^2-\delta^2)/(2\gamma+2\delta),0)$. $O_2=(-\delta,0)$ also lies on this circle. As shown in Fig. \ref{pic_division_bounded}(a), we use $b_1=0$ to denote this boundary, which is an arc inside region $\large{\textcircled{\small{3}}}\cup\large{\textcircled{\small{5}}}$.  Boundary $b_1=0$ subdivides region $\large{\textcircled{\small{3}}}\cup\large{\textcircled{\small{5}}}$ into two new regions: region $\large{\textcircled{\small{3}}}'$ lies outside the circle and region $\large{\textcircled{\small{5}}}'$ lies inside the circle. If $p$ lies in region $\large{\textcircled{\small{3}}}'$, $\lambda^*<\delta$, thus the optimal solution $x^*$ lies on $E^\circ(x_i,x_j)$. The distance from $x^*$ to $O_1$ equals to $\lambda^*$. If $p$ lies in region $\large{\textcircled{\small{5}}}'$, $\lambda^*>\delta$, thus the optimal solution is $O_2$.

When $p$ lies in region $\large{\textcircled{\small{2}}}\cup\large{\textcircled{\small{4}}}$, $f(\lambda)$ decreases in $(0,\lambda^*)$ and then increases in $(\lambda^*,+\infty)$. Thus if $f(0)>f(\delta)$ the optimal solution of $f(\lambda)$ is 0, and if $f(0)<f(\delta)$ the optimal solution of $f(\lambda)$ is $\delta$ . Therefore region $\large{\textcircled{\small{2}}}\cup\large{\textcircled{\small{4}}}$ could also be subdivided by a new boundary $f(0)=f(\delta)$, whose squared form is,
\begin{equation*}
\frac{\rho^2}{r^2} = \frac{\rho^2+2\rho\delta\cos\theta+\delta^2}{r^2+2\gamma\delta+\delta^2}.
\end{equation*}
Similar to the process of obtaining boundary $b_1=0$, the above equation can be expressed as
\begin{equation}
\label{eq_boundary2}
\quad (2\gamma+\delta)(p_x^2+p_y^2) - 2r^2p_x - r^2\delta = 0,
\end{equation}
which also represents a circle passing through $x_i$ and $x_j$. But the circle center is $(r^2/(2\gamma+\delta),0)$. As shown in Fig. \ref{pic_division_bounded}(a), we use $b_2=0$ to denote the boundary, which is an arc inside region $\large{\textcircled{\small{2}}}\cup\large{\textcircled{\small{4}}}$. Actually boundary $b_2=0$ must lie inside region $\large{\textcircled{\small{4}}}$, because the center lies on the positive x-axis. Boundary $b_2=0$ subdivides region $\large{\textcircled{\small{2}}}\cup\large{\textcircled{\small{4}}}$ into two new regions: region $\large{\textcircled{\small{2}}}'$ lies outside the circle and region $\large{\textcircled{\small{4}}}'$ lies inside the circle. If $p$ lies in region $\large{\textcircled{\small{2}}}'$, $f(0)>f(\delta)$, the optimal solution is $O_1$. If $p$ lies in region $\large{\textcircled{\small{4}}}'$, $f(0)<f(\delta)$, the optimal solution is $O_2$.

As a summary, the changes of the objective function and the corresponding optimal solutions for each region are listed in Table \ref{table_trend_bounded}. From the table we could also see that different regions may correspond to the same optimal solution, the six regions could also be degenerated into three. As shown in Fig. \ref{pic_division_bounded}(b), the boundaries of the three regions are three arcs $c=0\wedge b<0$, $b_1=0$ and $b_2=0$. We also use $E:c=0$, $E:b_1=0$, $E:b_2=0$ to denote the three boundaries.

When $O_1$ lies on segment $x_ix_j$ ($\gamma=0$), boundary $E:c=0$ lies on the line passing through $x_i$ and $x_j$. As shown in Fig. \ref{pic_division_bounded}(c), the plane still could be divided into three regions. The optimal solutions of each region correspond to either one endpoint or the interior of the bounded edge. Actually, the case of unbounded edge could be seen as a special case of bounded one, because boundary $E:b_1=0$ converges to $E:a=0$ and boundary $E:b_2=0$ converges to $E:b=0$ when $\delta\rightarrow \infty$.

\begin{figure}
  \centering
  \includegraphics[width=0.95\linewidth]{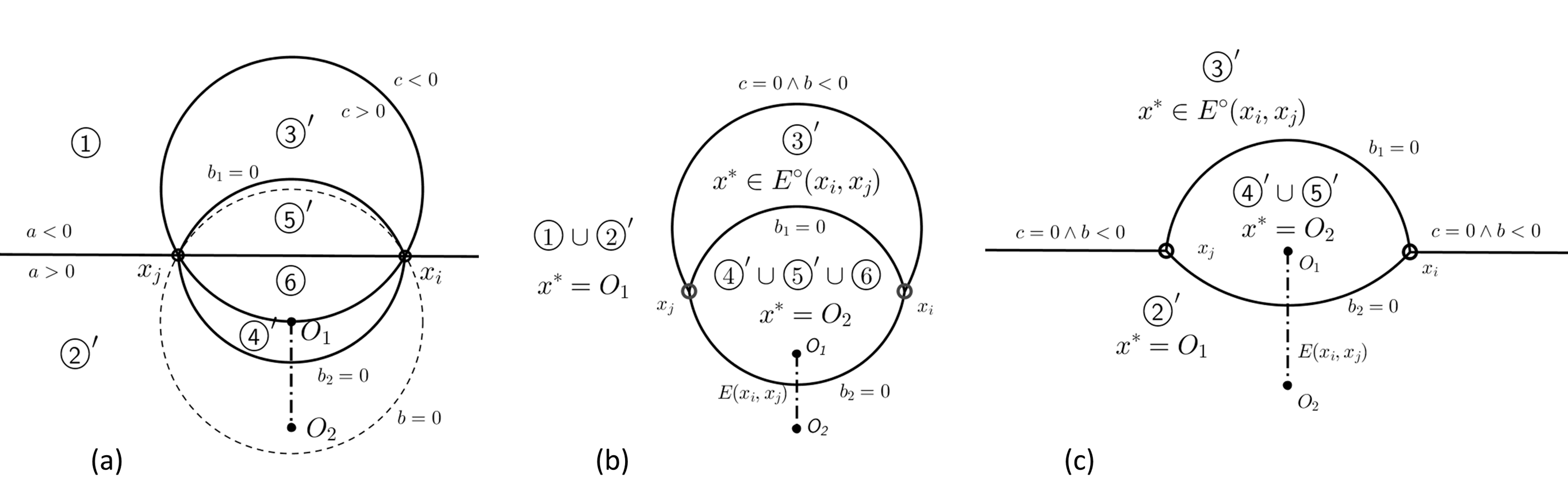}
  \caption{Plane division for a single bounded edge. (a) Original, $\gamma>0$, (b) Degenerate, $\gamma>0$, (c) Degenerate, $\gamma=0$}
  \label{pic_division_bounded}
\end{figure}

\begin{table}
  \caption{Optimal solution on a single bounded edge}
  \label{table_trend_bounded}
  \small
  \renewcommand{\arraystretch}{1.5}
  \centering
  \begin{tabular}{|*{5}{c|}}
  \hline
  Region & $f(\lambda)$ & Value Comparison & Optimal value & Optimal Solution \\ \hline
  $\large{\textcircled{\small{1}}}$   & $(0,\delta)\searrow$  & $f(0)>f(\delta)$ & $f(0)$ & $O_1$   \\ \hline
  $\large{\textcircled{\small{2}}}'$  & $(0,\delta)\searrow$ or $(0,\lambda^*)\searrow(\lambda^*,\delta)\nearrow$ & $f(0)>f(\delta)$ & $f(0)$ & $O_1$ \\ \hline
  $\large{\textcircled{\small{3}}}'$  & $(0,\lambda^*)\nearrow(\lambda^*,\delta)\searrow$ & $f(\lambda^*)>\max(f(0),f(\delta))$ & $f(\lambda^*)$ & $\in E^\circ(x_i,x_j)$ \\ \hline
  $\large{\textcircled{\small{4}}}'$  & $(0,\lambda^*)\searrow(\lambda^*,\delta)\nearrow$ & $f(0)<f(\delta)$ & $f(\delta)$ & $O_2$ \\ \hline
  $\large{\textcircled{\small{5}}}'$  & $(0,\delta)\nearrow$ & $f(0)<f(\delta)$ & $f(\delta)$  &$O_2$ \\ \hline
  $\large{\textcircled{\small{6}}}$   & $(0,\delta)\nearrow$ & $f(0)<f(\delta)$ & $f(\delta)$  &$O_2$ \\ \hline
  \end{tabular}
\end{table}

\subsection{summary}
In the above discussion, no matter whether the edge is bounded or not, and no matter whether the starting node lies on segment $x_ix_j$ or not, the plane could be divided into three regions. When $p$ lies in one of the regions, the optimal solution of subproblem $M_s(W_4)$ either locates at one endpoint of $E(x_i,x_j)$ or lies on $E^\circ(x_i,x_j)$. 
From Table \ref{table_trend_unbounded} and Table \ref{table_trend_bounded}, the objective function value at the optimal solution on $E^\circ(x_i,x_j)$ is larger than that at two endpoints of $E(x_i,x_j)$.
\begin{defn}
We call the region whose corresponding optimal solution lies on $E^\circ(x_i,x_j)$ as the dominant region of $E(x_i,x_j)$. And we call the regions corresponding to two endpoints of $E(x_i,x_j)$ as the adjoint regions of $E(x_i,x_j)$.
\end{defn}
We denote the dominant region of $E(x_i,x_j)$ by $DR(E)$ and denote the two adjoint regions by $AR(O_1)$ and $AR(O_2)$. What should be noticed is the boundary between two adjoint regions. If $E(x_i,x_j)$ is unbounded and $p$ lies on boundary $E:b=0$, we have $f(0)=f(\infty)$, both $O_1$ and $\infty$ are the optimal solution. If $E(x_i,x_j)$ is bounded and $p$ lies on boundary $E:b_2=0$, we have $f(0)=f(\delta)$, both $O_1$ and $O_2$ are the optimal solution. Thus optimal solution on a single edge in $FVB(S)$ may not be unique.

Dominant region plays an important role in calculating the optimal solution on multiple edges. When $E(x_i,x_j)$ is unbounded, $DR(E)$ is an arch if $\gamma>0$ and $DR(E)$ is a half plane if $\gamma=0$. When $E(x_i,x_j)$ is bounded, $DR(E)$ is the difference set of two arches if $\gamma>0$ and $DR(E)$ is the difference set of one half plane and one arch if $\gamma=0$. In order to represent the dominant region of arbitrary edge in a unified form, we give an extended definition of arch region to include the special case in which $\gamma=0$.
\begin{defn}
\label{definition_arch}
$O$ is the starting node of $E(x_i,x_j)$. If $O$ does not lie on segment $x_ix_j$ ($\gamma>0$), segment $x_ix_j$ divides the circle passing through $O,x_i$ and $x_j$ into two arcs. Denote the arc which does not contain point $O$ by $Arc(O,x_i,x_j)$, then  $Arch(O,x_i,x_j)$ is the region surrounded by segment $x_ix_j$ and $Arc(O,x_i,x_j)$. If $O$ lies on segment $x_ix_j$ ($\gamma=0$), the line passing the three points divides the plane into two half planes. Then $Arch(O,x_i,x_j)$ indicates the half plane which does not contain $E(x_i,x_j)$.
\end{defn}
From the definition, we have $Arch(\infty,x_i,x_j)=\emptyset$, thus the dominate region of an arbitrary edge could be represented as,
\begin{equation}
DR(E) = Arch(O_1,x_i,x_j) - Arch(O_2,x_i,x_j).
\end{equation}

\section{Optimal solution on the whole $FVB(S)$}
\label{section_center_function}

After the optimal solutions on each single edge are calculated, the optimal solution on the whole $FVB(S)$ could be obtained by traversing all edges in $FVB(S)$. Assume $CH(S)$ has $m$ distinct vertex, then there are at most $2m-3$ edges in $FVB(S)$ \cite{book2008CG}. Even if the edge containing $\varepsilon(S)$ is divided, the edge number in $FVB(S)$ is no larger than $2m-2$. Thus the time complexity of the above numerical method is $O(m)$ if $FVB(S)$ is known. Both $FVB(S)$ and $\varepsilon(S)$ can be calculated using the randomized incremental algorithm \cite{book2008CG}, whose time complexity is also $O(m)$.

However, several questions still remain unanswered: Is the optimal solution of problem $M(W_4)$ unique? Do all edges in $FVB(S)$ need to be checked? In order to answer the two questions, we present a tree structure constructed from $FVB(S)$. The hierarchical relationship between edges in the tree will help us to locate the optimal solution.
\begin{defn}
Treat $FVB(S)$ as a graph $G=(V,E)$ with the node set $V$ and the edge set $E$. $FVB(S)$ forms a tree-like structure. If $\varepsilon(S)$ locates at one node in $FVB(S)$, we use the node as the root node. If $\varepsilon(S)$ lies on the interior of one edge in $FVB(S)$, we separate the edge into two parts from $\varepsilon(S)$ and treat $\varepsilon(S)$ as the root node. The tree constructed from $FVB(S)$ using $\varepsilon(S)$ as the root node is called the division tree. 
\end{defn}
Denote the division tree associated with point set $S$ by $DV(S)$, and denote the depth of $DV(S)$ by $D$. The leaf node of $DV(S)$ locates at $\infty$. Edges in $DV(S)$ connect one node with its child node, thus we define the edge depth as the depth of the connected child node. All edges and nodes with depth less than or equal to $d$ also form a tree structure $(d\leq D)$, which is denoted by $DV_d(S)$. If one leaf node of $DV_d(S)$ does not locate at $\infty$, we extend the edge connecting the leaf node in $DV_d(S)$ to infinity along the original direction of the edge, and denote the extended tree by $DV_d'(S)$. Then $DV_d'(S)$ is the division tree associated with a subset of $S$, which contains all static points involved in all edges in $DV_d'(S)$. Denote the subset by $S_d$, then we have $DV_d'(S)=DV(S_d)$.

\begin{figure}
  \centering
  \includegraphics[width=1\linewidth]{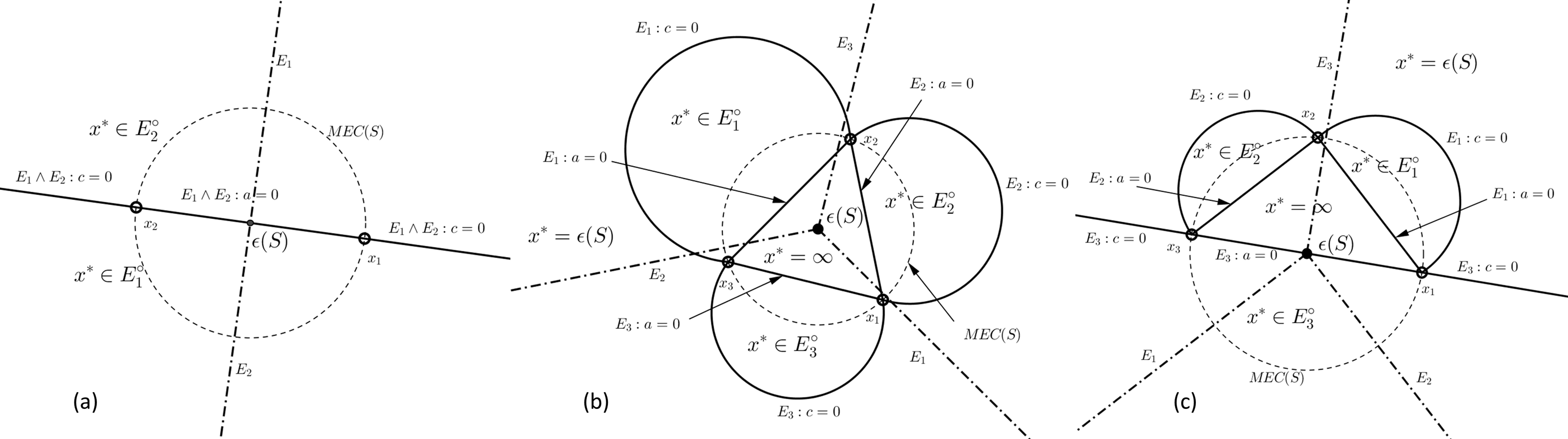}\\
  \caption{Plane division and corresponding optimal solution of the three basic types of $DV_1'(S)$. (a) Two points on a diameter, (b) Three points on an acute-angled triangle, (c) Three points on a right-angled triangle.}
  \label{pic_os_depth1}
\end{figure}

\begin{Lemma}
\label{lemma_start}
The optimal solution of problem $M(W_4)$ on $DV(S_1)$ is unique. The plane could be divided into several non-overlapping regions. When the weight point lies in one of the regions, the optimal solution either locates at one node or lies on the interior of one edge.
\end{Lemma}
\begin{proof}
If the depth of the division tree is 1, all static points in $S$ lie on $MEC(S)$. Since $S$ is in general position, there are at most three static points in $S$. Thus the point configuration of $S$ may only have three basic types: (1) two points $x_1,x_2$ on a diameter, (2) three points $x_1,x_2,x_3$ form an acute-angled triangular, and (3) three points $x_1,x_2,x_3$ form a right-angled triangle. As shown in Fig. \ref{pic_os_depth1}, in the first type $\varepsilon(S)$ lies on segment $x_ix_j$, in the second type $\varepsilon(S)$ lies inside $\triangle x_1x_2x_3$, and in the third type $\varepsilon(S)$ lies on the right-angle side of $\triangle x_1x_2x_3$. All edges starting from the root node $\epsilon(S)$ are unbounded.

Take the second type as the example. The plane is divided into five non-overlapping regions: the dominant regions of three edges $DR(E_1)$, $DR(E_2)$ and $DR(E_3)$ correspond to the interiors of the three edges $E_1^\circ$, $E_2^\circ$ and $E_3^\circ$ respectively, $\triangle x_1x_2x_3$ corresponds to $\infty$, and the rest region corresponds to $\varepsilon(S)$. Because $x_1$, $x_2$ and $x_3$ are the possible intersection points of the boundaries between any two adjacent regions and the weight point $p$ is required to not coincide with the three static points in problem $M(W_4)$, the optimal solution is unique. It is easy to verify that the conclusions hold true for the other two basic types. 
\end{proof}

\begin{figure}
  \centering
  \includegraphics[width=0.7\linewidth]{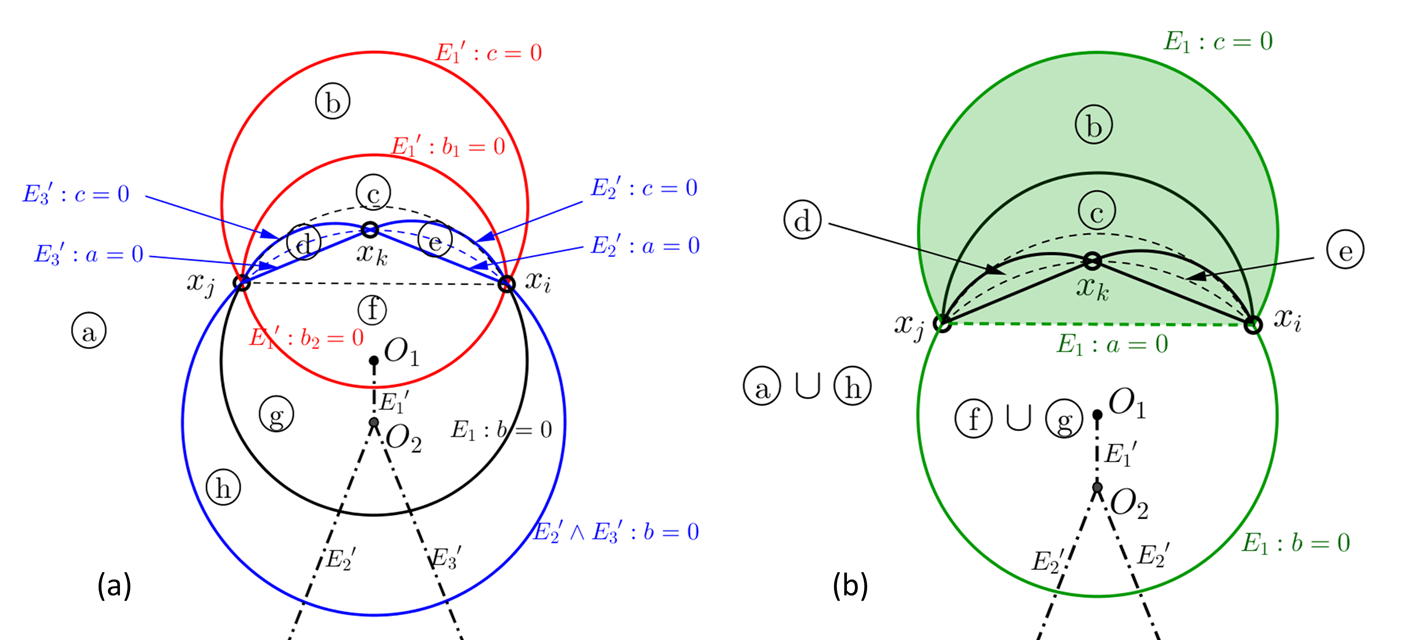}
  \caption{Plane division for a subtree in $DV_{d+1}'(S)$. (a) Original, $\gamma>0$, (b) Degenerate, $\gamma>0$.}
  \label{pic_cf_second_primary_case}
\end{figure}

\begin{Lemma}
\label{lemma_deductive}
Assume the conclusions of Lemma \ref{lemma_start} hold true for $DV(S_d)$ $(0\leq d < D-1)$. In order to obtain the plane division for $DV(S_{d+1})$ and the corresponding optimal solution of each region, we check all unbounded edges in $DV_d'(S)$ with depth equal to $d$. If an unbounded edge $E_1(x_i,x_j)$ in $DV(S_d)$ is split into two unbounded edges $E'_2(x_i,x_k)$ and $E'_3(x_k,x_j)$ at node $O_2$ in $DV(S_{d+1})$, we subdivide $Arch(O_1,x_i,x_j)$, the dominant region of $E_1$, into non-overlapping subregions. The corresponding unique optimal solution of each subregion is,
\begin{equation}
\left\{
\begin{array}{lcl}
x^*=\infty,          & & p \in \triangle x_ix_jx_k  \\
x^*\in E_2'^\circ(x_i,x_k),  & & p \in Arch(O_2,x_i,x_k)  \\
x^*\in E_3'^\circ(x_k,x_j),  & & p \in Arch(O_2,x_k,x_j)  \\
x^*\in E_1'^\circ(x_i,x_j) & & p \in Arch(O_1,x_i,x_j)-Arch(O_2,x_i,x_j) \\
x^*=O_2,             & & p \in remaining\ region\ in\ Arch(O_1,x_i,x_j)
\end{array}\right..
\label{eq_centerfunction_altered}
\end{equation}
And the plane division and corresponding unique optimal solutions for other regions remain unchanged.
\end{Lemma}
\begin{proof}
The difference between $DV(S_d)$ an $DV(S_{d+1})$ is that one unbounded edge with depth $d$ in $DV(S_d)$ may be replaced by a subtree containing a bounded edge with depth $d$ and two unbounded edges with depth $d+1$ in $DV(S_{d+1})$.

We investigate the plane division and the corresponding optimal solution on one of the subtrees in $DV(S_{d+1})$. As shown in Fig. \ref{pic_cf_second_primary_case}(a), the subtree contains three edges: $E_1'(x_i,x_j)$ is bounded, $E_2'(x_i,x_k)$ and $E_3'(x_k,x_j)$ are unbounded. $O_2$ is the intersection point of three edges and $O_1$ is the parent node of $O_2$. $O_1$ is not on segment $x_ix_j$ ($\gamma>0$). From Fact \ref{fact_adjointcircle}, $x_i$, $x_j$ and $x_k$ should lie on a circle centered at $O_2$.  We obtain the plane division associated with $E_1'$ and $E_2'\cup E_3'$ separately and then calculate the corresponding optimal solutions of each new region. The plane division associated with the bounded edge $E_1'$ is shown with the red solid line. The plane division associated with two unbounded edges $E_2'\cup E_3'$ is shown with the blue solid line. The above two groups of boundaries together with $E_1':b=0$ divide the plane into eight regions. The corresponding optimal solutions of each region are listed in Table \ref{table_subtree}. From the table we could see that different regions may correspond to the same optimal solution, thus the above eight regions could be degenerated into six, which is shown in Fig. \ref{pic_cf_second_primary_case}(b). On the same image, the plane division associated with the unbounded edge $E_1$ in $DV(S_d)$ is shown with the green lines. Through comparison of the two plane divisions we have the following conclusion: In order to obtain the plane division associated with the subtree, the dominant region of the unbounded edge $E_1$, namely $Arch(O_1,x_i,x_j)$, needs to be subdivided according to Eq. \eqref{eq_centerfunction_altered} and the two adjoint regions of $E_1$ remain unchanged. When $O_1$ is on segment $x_ix_j$ ($\gamma=0$), the above conclusion is still true with the extended definition of the arch region.

As $Arch(O_1,x_i,x_j)$ is also one of the non-overlapping regions in the plane division associated with $DV(S_d)$, the regions to be subdivided do not overlap during the check of different unbounded edges in $DV(S_d)$ with depth $d$. After all unbounded edges are dealt with, the plane division for $DV(S_{d+1})$ is obtained. The unique optimal solution corresponding to the new subregions can be calculated according to Eq. \eqref{eq_centerfunction_altered} and the unique optimal solution corresponding to other regions remains unchanged.

\end{proof}
\begin{table}
  \caption{Optimal solution on one subtree in $DV_{d+1}'(S)$}
  \label{table_subtree}
  \small
  \renewcommand{\arraystretch}{1.5}
  \centering
  \begin{tabular}{|*{5}{c|}}
  \hline
  Region  &  $x^*_{E_1'}$ &   $x^*_{E_2'\cup E_3'}$ & Value comparison  &  $x^*_{ST}$           \\ \hline
  $\large{\textcircled{\small{a}}}$   & $O_1$         & $O_2$             & $f(O_1)>f(O_2)$   & $O_1$                  \\ \hline
  $\large{\textcircled{\small{b}}}$   & $\in E_1'^\circ$     & $O_2$             & $f(x^*_{E_1'})>f(O_2)$   & $\in E_1'^\circ$           \\ \hline
  $\large{\textcircled{\small{c}}}$   & $O_2$         & $O_2$             & -                 & $O_2$                 \\ \hline
  $\large{\textcircled{\small{d}}}$   & $O_2$         & $\in E_3'^\circ$   & $f(x^*_{E_2'\cup E_3'})>f(O_2)$   & $\in E_3'^\circ$      \\ \hline
  $\large{\textcircled{\small{e}}}$   & $O_2$         & $\in E_2'^\circ$   & $f(x^*_{E_2'\cup E_3'})>f(O_2)$   & $\in E_2'^\circ$       \\ \hline
  $\large{\textcircled{\small{f}}}$   & $O_2$         & $\infty$          & $f(\infty)>f(O_2)$& $\infty$                \\ \hline
  $\large{\textcircled{\small{g}}}$   & $O_1$         & $\infty$          & $f(\infty)>f(O_1)$& $\infty$               \\ \hline
  $\large{\textcircled{\small{h}}}$   & $O_1$         & $\infty$          & $f(O_1)>f(\infty)$& $O_1$                   \\ \hline
  \end{tabular}
\end{table}

According to Lemma \ref{lemma_start} and Lemma \ref{lemma_deductive}, we have the following conclusions through induction.
\begin{Theorem}
\label{theorem_arbitrary_case}
The optimal solution of problem $M(W_4)$ on arbitrary $DV(S)$ is unique. The plane could be divided into non-overlapping regions, when the weight point lies in one of the regions, the optimal solution either locates at one node or lies on the interior of one edge in $DV(S)$.
\end{Theorem}

Now that the optimal solution of problem $M(W_4)$ is unique, we could describe the relationship between the weight point and the optimal solution with the concept of center function.
\begin{defn}
Function $\varphi:p\rightarrow x^*$ maps the weight point $p$ to the optimal solution $x^*$ of problem $M(W_4)$. Since $x^*$ is also the center of the minimum enclosing circle with a dynamic weight, we call $\varphi:p\rightarrow x^*$ as the center function and denote the center function by $\varphi(p)$.
\end{defn}
Correspondingly, we could define the inverse image of the center function.
\begin{defn}
$q$ is a point in $DV(S)$ and $Q$ is a subset of $DV(S)$. Then both $\varphi^{-1}(q)$ and $\varphi^{-1}(Q)$ are subsets of the Euclidean plane. If $p$ belongs to $\varphi^{-1}(q)$, the optimal solution of problem $M(W_4)$ locates at $x^*$. If $p$ belongs to $\varphi^{-1}_E(Q)$, the optimal solution of problem $M(W_4)$ belongs to set $Q$.
\end{defn}
In problem $M(W_4)$, the weight point $p$ should not coincide with any vertex of $CH(S)$, thus the domain of $\varphi(p)$ is the Euclidean plane except for the vertex of $CH(S)$. From Theorem \ref{theorem_location}, $FVB(S)$ is the range of $\varphi(p)$. When the weight point lies in $CH(S)$, the optimal solution is $\infty$, thus $\varphi^{-1}(\infty)=CH(S)-S$. When the weight point lies in the dominant region of one edge $E$, the optimal solution lies on $E^\circ$, thus we have $\varphi^{-1}(E^\circ)=DR(E)$. The distance from the optimal solution $x^*$ to the connected parent node is $\lambda^*$, which could be calculated according to Eq. \eqref{eq_lambda} with the local coordinate system shown in Fig. \ref{pic_unbounded_edge}. The inverse image of node $O$ depends on whether the node is the root node or not. If $O=\varepsilon(S)$, $\varphi^{-1}(O)$ can be obtained according to the three basic types in Lemma \ref{lemma_start} which are shown in Fig. \ref{pic_os_depth1}. If $O\neq\varepsilon(S)$, we assume node $O$ is the intersection of three edges $E_1(x_i,x_j)$, $E_2(x_i,x_k)$ and $E_3(x_j,x_k)$, then $\varphi^{-1}(O)$ is the region surrounded by the arcs corresponding to the three edges, namely $Arc(O,x_i,x_j)$, $Arc(O,x_i,x_k)$ and $Arc(O,x_j,x_k)$. Fig. \ref{pic_cf_twonodes} shows the center function of two cases in which there are two nodes in $FVB(S)$. In the first case, $\varepsilon(S)$ locates at one node in $FVB(S)$. In the second case, $\varepsilon(S)$ lies on the interior of one edge in $FVB(S)$.
\begin{figure}
    \centering
    \includegraphics[width=0.85\linewidth]{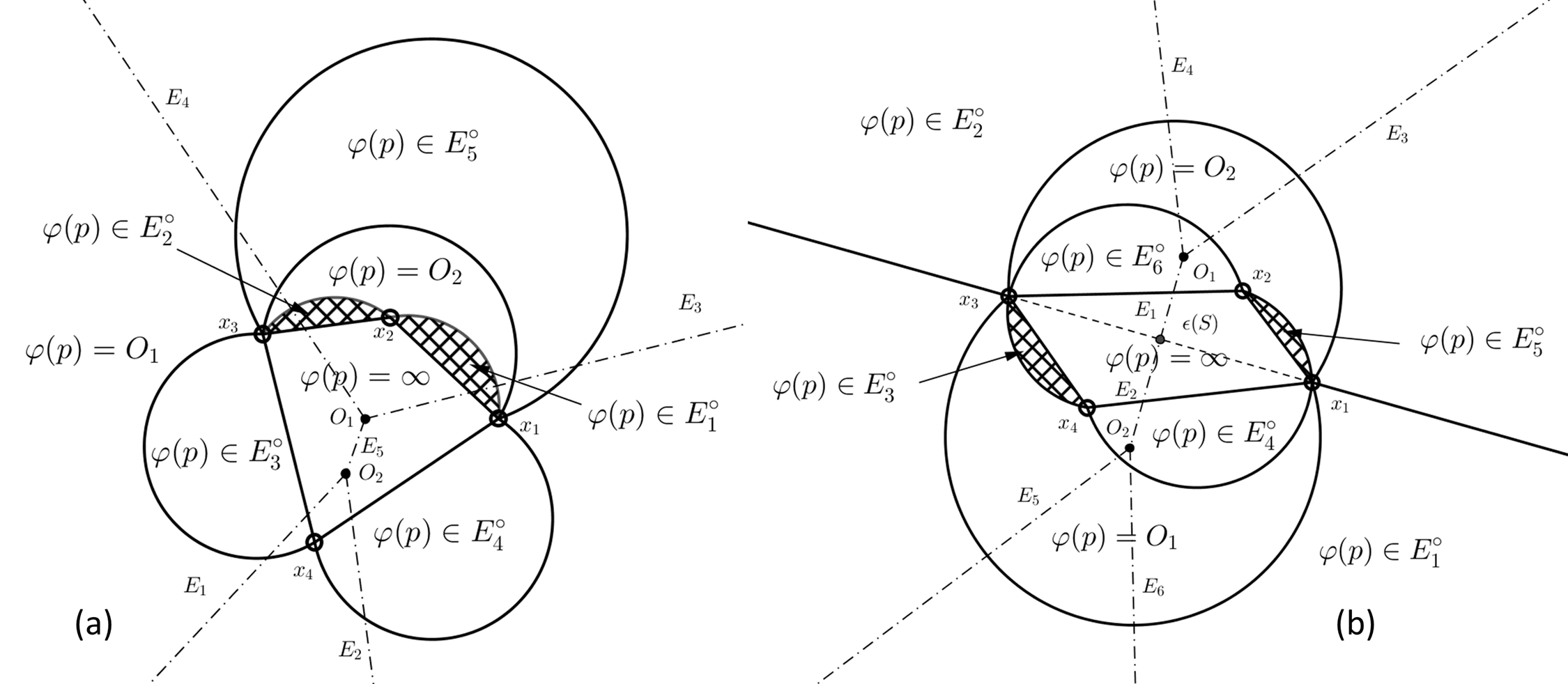}
    \caption{Center function of two cases in which there are two nodes in $FVB(S)$. (a) $\varepsilon(S)$ locates at node $O_1$, (b) $\varepsilon(S)$ lies on the interior of $O_1O_2$.}
    \label{pic_cf_twonodes}
\end{figure}

\begin{Theorem}
Center function is continuous when $p\notin CH(S)$.
\end{Theorem}
\begin{proof}
When $p\notin CH(S)$, $\varphi(p)$ is defined and $\varphi(p)<\infty$, the continuity could be discussed. If $p$ lies in the region corresponding to a node in $DV(S)$, the center function is constant, thus it is continuous. If $p$ lies in the dominant region of an edge $E(x_i,x_j)$ in $DV(S)$, the optimal solution $x^*$ must lie on $E^\circ(x_i,x_j)$. From the presentation of boundary $b_1=0$ in Section \ref{subsection_bounded_edge}, $\varphi^{-1}(x^*)$ should be an arc lying on the circle which passes through $x_i$, $x_j$ and $x^*$. With $x^*$ moving along $E(x_i,x_j)$ from the starting node $O_1$ to the ending node $O_2$, $\varphi^{-1}(x^*)$ changes from boundary $E:c=0$ to $E:b_1=0$ continuously, thus the center function is continuous in $DR(E)$. Similarly, the center function is also continuous on any boundary of two adjacent regions.
\end{proof}

\begin{Theorem}
Assume $CH(S)$ has $m$ distinct vertex, then there are at most $3m-4$ non-overlapping regions in the plane division associated with $FVB(S)$.
\end{Theorem}
\begin{proof}
There are at most $2m-3$ edges and $m-1$ nodes ($\infty$ is treated as one node) in $FVB(S)$. During the construction of $DV(S)$ from $FVB(S)$, if $\varepsilon(S)$ lies on the interior of one edge, like $E(x_1,x_2)$, there are one more edge and one more node in $DV(S)$ than $FVB(S)$ after the separation of the edge, thus there should be at most $3m-2$ non-overlapping regions in the plane division associated with $DV(S)$. However, $\varphi^{-1}(\varepsilon(S))$ is the boundary between the dominant regions of the two separated edges $\varphi^{-1}(E_1)$ and $\varphi^{-1}(E_2)$ according to the plane division of the first type in Lemma \ref{lemma_start}. When the weight point lies on one of the three regions, the optimal solution must lie on $E^\circ(x_1,x_2)$. Thus the three regions are combined as one in the plane division associated with $FVB(S)$.
\end{proof}

Given the division tree and an arbitrary weight point, we just need to determine which region the weight point lies in to calculate the optimal solution. From Lemma \ref{lemma_deductive}, if the weight point $p$ has been verified to lie in the dominant of an unbounded edge $E_1$ in $DV(S_d)$ with depth $d$, we search the corresponding subtree in $DV(S_{d+1})$ and skip other branches. Thus not all edges in the division tree need to be searched and the time complexity in worst cases is consistent with the depth of the division tree. For balanced tree structure, $D\sim O(\log m)$.

\section{Application}
In this section we present a practical problem which could be solved using the proposed new variation. Assume all distinct points in $S=\{x_1,x_2,...,x_n\}$ $(n\geq 2)$ lie on a rigid object in 2D space and restrict the maximum displacement in 2D rigid motion of points in $S$, we want to know the possible maximum displacement of another point $p$ on the rigid object in 2D rigid motion. The displacement of certain point after 2D rigid motion could be represented by the Target Registration Error (TRE) \cite{Fitzpatrick1998}.
\begin{defn}
The Target Registration Error (TRE) of point $p$ in $\mathbb{R}^2$ is defined as $TRE(p)=Rp+s-p$, where $R$ is the rotation matrix and $s$ is the translation term.
\end{defn}
Usually the rotation matrix is denoted by $R=[\cos\theta, -\sin\theta; \sin\theta, \cos\theta]$, where $\theta$ is the rotation angle, and the translation term is denoted by $s=[s_x;s_y]$. Use $\|\cdot\|$ to denote the magnitude of $TRE$. then the above problem could be represented as,
\begin{equation}
\label{eq_perturbation}
M(P):\quad \max_{R,s} \|TRE(p)\| \qquad  s.t.\quad \|TRE(x_i)\|\leq C, \forall\ x_i\in S
\end{equation}
where $C$ is the maximum displacement of points in $S$. Both the objective function and the constraint function are the $TRE$ magnitude of some point, whose characteristics are revealed by the following two lemmas.
\begin{Lemma}
\label{lemma_contourline}
If $\theta\neq 0$, the contour line of the $TRE$ magnitude in $\mathbb{R}^2$ is a circle. The $TRE$ magnitude of arbitrary point $p$ is,
\begin{equation}
\|TRE(p)\|=2l\sin\frac{|\theta|}{2},
\end{equation}
where $l$ is the $l_2$ distance from $p$ to the circle center $(c_x,c_y)$.
 which is,
\begin{equation}
\label{eq_circle}
c_x = \frac{1}{2}(s_x - s_ycot\frac{\theta}{2}), \quad c_y = \frac{1}{2}(s_xcot\frac{\theta}{2}+s_y).
\end{equation}
\end{Lemma}
\begin{proof}
Denote point $p$ by the Cartesian coordinate $(p_x,p_y)$, then the $TRE$ of point $p$ could be expressed as
\begin{equation*}
TRE(p) = [p_xcos\theta-p_ysin\theta-p_x+s_x;\ p_xsin\theta+p_ycos\theta-p_y+s_y].
\end{equation*}
Set $\|TRE(p)\|^2=v^2$, where $v$ is the magnitude of $TRE$, we have
\begin{equation*}
(p_xcos\theta-p_ysin\theta-p_x+s_x)^2 +(p_xsin\theta+p_ycos\theta-p_y+s_y)^2=v^2.
\end{equation*}
Combine like terms of $p_x$ and $p_y$,
\begin{equation*}
(2-2cos\theta)(p_x^2+p_y^2) + 2p_x(s_x(cos\theta-1)+s_ysin\theta) + 2p_y(-s_xsin\theta+s_y(cos\theta-1)) + s_x^2 + s_y^2 = v^2.
\end{equation*}
The above equation is a circle, and the circle center could be represented using Eq. \eqref{eq_circle}.
Then the standard equation of the circle could be represented,
\begin{equation*}
(p_x-c_x)^2 + (p_y-c_y)^2 = v^2/(4sin^2\frac{\theta}{2}).
\end{equation*}
\end{proof}
\begin{Lemma}
\label{lemma_contourline_2}
If $\theta=0$, $\|TRE(p)\|=\|s\|$ for any point in the plane.
\end{Lemma}
According to the above two lemmas, we have the following conclusion.
\begin{Theorem}
When $0<C\leq 2r(S)$, problem $M(W_1)$ is equivalent with problem $M(P)$.
\end{Theorem}
\begin{proof}
If $\theta\neq0$, problem $M(P)$ could be represented as below according to Lemma \ref{lemma_contourline},
\begin{equation*}
\max_{x,\theta} 2\|x-p\|\sin\frac{|\theta|}{2}  \qquad s.t. \quad \forall i,\ 2\|x-x_i\|\sin\frac{|\theta|}{2}\leq C,
\end{equation*}
where $x$ is the undetermined circle center of the contour line. From the constraint we have,
\begin{equation*}
\sin\frac{|\theta|}{2}\leq \min_{i}\frac{C}{2\|x-x_i\|} = \frac{C}{2\max_{i}\|x-x_i\|}.
\end{equation*}
When $0<C\leq 2r(S)$, $\forall x\in\mathbb{R}^2$, there exists a $\theta$ such that the equality of the constraint holds. Thus the objective function could be represented as,
\begin{equation}
\label{eq_equivalent}
\max_{x} \frac{C\|x-p\|}{\max_i \|x-x_i\|}.
\end{equation}
which is only different from problem $M(W_1)$ with a constant $C$.

If $\theta=0$, problem $M(P)$ could be represented as below according to Lemma \ref{lemma_contourline_2},
\begin{equation*}
\max_{s} \|s\| \qquad s.t. \quad \|s\|\leq C.
\end{equation*}
The optimal value is $C$, which is identical with the objective function value of Eq. \eqref{eq_equivalent} when $x=\infty$.
\end{proof}
Therefore problem $M(P)$ could be analytically solved using the new variation of the minimum enclosing circle problem.
\section{Conclusion}
In this article we studied a new variation of the minimum enclosing circle problem, in which the distance from the undetermined circle center to the weight point $p$ is used as the dynamic weight. The time complexity of each step in solving the problem is analysed below. Given a set $S$ of $n$ static points, $CH(S)$ could be calculated in $O(n\log n)$ time. Assume there are $m$ distinct vertex in $CH(S)$, then $FVB(S)$ and $\varepsilon(S)$ could be calculated in $O(m)$ time. $FVB(S)$ and $\varepsilon(S)$ form the division tree. Given the division tree and an arbitrary weight point, the time complexity for worst cases is consistence with the depth of the division tree. For balanced tree structure, $D\sim O(\log m)$. The future work may extend to higher space $\mathbb{R}^d (d>2)$ or study the cases when both static weight $w_i$ and dynamic weight $w(x)=1/\|x-p\|$ are used at the same time.
\section*{Acknowledgement}
This work is partly supported by the National Natural Science Foundation of China with Grant No. 61172125 and No. 61132007.
\section*{Reference}
\bibliographystyle{elsarticle-num}

\begin{thebibliography}{10}
\bibitem{sylvester1857question}
J.~J. Sylvester, {A question in the geometry of situation}, Quarterly Journal
  of Pure and Applied Mathematics 1 (1857).

\bibitem{megiddo1983linear}
N.~Megiddo, {Linear-time algorithms for linear programming in $R^{3}$ and
  related problems}, SIAM journal on computing 12~(4) (1983) 759--776.

\bibitem{Welzl1991}
E.~Welzl, {Smallest
  enclosing disks (balls and ellipsoids)}, New Results and New Trends in
  Computer Science 11~(3) (1991) 359--370.

\bibitem{dearing2013dual}
P.~M. Dearing, A.~M. Smith, {A dual algorithm for the minimum covering weighted
  ball problem in $\mathbb{R}^n$}, Journal of
  Global Optimization 55~(2) (2013) 261--278.

\bibitem{hearn1982efficient}
D.~W. Hearn, J.~Vijay, {Efficient algorithms for the (weighted) minimum circle
  problem}, Operations Research 30~(4) (1982) 777--795.

\bibitem{megiddo1983weighted}
N.~Megiddo, {The weighted Euclidean 1-center problem}, Mathematics of
  Operations Research 8~(4) (1983) 498--504.

\bibitem{Aggarwal1989}
A.~Aggarwal, L.~J. Guibas, J.~Saxe, P.~W. Shor, {A linear-time algorithm for
  computing the voronoi diagram of a convex polygon}, Discrete {\&}
  Computational Geometry 4~(1) (1989) 591--604.

\bibitem{chazelle1996linear}
B.~Chazelle, J.~Matou{\v{s}}ek, {On linear-time deterministic algorithms for
  optimization problems in fixed dimension}, Journal of Algorithms 21~(3)
  (1996) 579--597.

\bibitem{Banik2014}
A.~Banik, B.~B. Bhattacharya, S.~Das,
  {Minimum enclosing circle of a set of fixed points and a mobile point}, Computational Geometry:
  Theory and Applications 47~(9) (2014) 891--898.

\bibitem{book2008CG}
M.~de~Berg, O.~Cheong, M.~van Kreveld, M.~Overmars, {Computational Geometry:
  Algorithms and Applications}, Springer (2008).

\bibitem{Fitzpatrick1998}
J.~Michael~Fitzpatrick, Jay~B.~West and Calvin~R.~Maurer, {Predicting Error in Rigid-Body Point-Based Registration}, IEEE Transactions on Medical Imaging 17~(5) (1998) 694--702.
\end{thebibliography}

\end{document}